\documentclass[amsmath,secnumarabic,floatfix,amssymb,nofootinbib,nobibnotes,letterpaper,11pt,tightenlines]{revtex4}

\usepackage{times}

\usepackage{geometry}
\usepackage{amssymb}
\usepackage{latexsym, amsmath, amscd,amsthm}
\usepackage{graphicx}
\usepackage[percent]{overpic}
\usepackage{units}
\usepackage{hyperref}
\PassOptionsToPackage{caption=false}{subfig}
\usepackage[lofdepth]{subfig}
\usepackage{booktabs}

\usepackage{clrscode}

\def\figdir{figs/}
\graphicspath{\figdir}

\newtheorem{theorem}{Theorem}

\newtheorem{lemma}[theorem]{Lemma}
\newtheorem{proposition}[theorem]{Proposition}
\newtheorem{corollary}[theorem]{Corollary}

\theoremstyle{definition}
\newtheorem{definition}[theorem]{Definition}

\newtheorem*{gFactor}{Theorem~\ref*{thm:contraction factors}}
\newtheorem*{degreeKirchhoff}{Theorem~\ref*{thm:degree Kirchhoff formula}}

\newcommand{\abs}[1]{\lvert#1\rvert}
\newcommand{\R}{\mathbb{R}}

\newcommand{\norm}[1]{\left\| #1 \right\|}

\newcommand{\grad}{\operatorname{grad}}

\newcommand{\graphG}{\mathbf{G}}
\newcommand{\edgesE}{\mathbf{e}}
\newcommand{\verticesV}{\mathbf{v}}
\newcommand{\edge}{e}
\newcommand{\vertex}{v}

\newcommand{\im}{\operatorname{im}}

\newcommand{\Loops}{\operatorname{Loops}}

\newcommand{\Kirchhoff}{\operatorname{Kf}}
\newcommand{\DegreeKirchhoff}{\operatorname{Kf}^*}
\newcommand{\tr}{\operatorname{tr}}
\renewcommand{\div}{\operatorname{div}}

\newcommand{\vertexa}{a}
\newcommand{\vertexb}{b}
\newcommand{\vertexw}{w}
\newcommand{\vertexv}{\vertex}
\newcommand{\vertexc}{c}
\newcommand{\vertexd}{d}
\newcommand{\vertexq}{q}
\newcommand{\vertexr}{r}

\def\co{\colon\thinspace}

\let\mgp=\marginpar \marginparwidth18mm \marginparsep1mm
\def\marginpar#1{\mgp{\raggedright\tiny #1}}

\let\lbl=\label
\def\label#1{\lbl{#1}\ifinner\else\marginpar{\ref{#1} #1}\ignorespaces\fi}

\begin{document}
\title[]{Radius of gyration, contraction factors, and subdivisions of topological polymers}
\author{Jason Cantarella}
\altaffiliation{Mathematics Department, University of Georgia, Athens, GA, USA}
\noaffiliation
\author{Tetsuo Deguchi}
\altaffiliation{Department of Physics, Ochanomizu University, Tokyo 112--8610, Japan}
\noaffiliation
\author{Clayton Shonkwiler}
\altaffiliation{Department of Mathematics, Colorado State University, Fort Collins, CO, USA}
\noaffiliation
\author{Erica Uehara}
\altaffiliation{Department of Physics, Ochanomizu University, Tokyo 112--8610, Japan}
\noaffiliation


\begin{abstract}
We consider the topologically constrained random walk model for topological polymers. In this model, the polymer forms an arbitrary graph whose edges are selected from an appropriate multivariate Gaussian which takes into account the constraints imposed by the graph type. We recover the result that the expected radius of gyration can be given exactly in terms of the Kirchhoff index of the graph. We then consider the expected radius of gyration of a topological polymer whose edges are subdivided into $n$ pieces. We prove that the contraction factor of a subdivided polymer approaches a limit as the number of subdivisions increases, and compute the limit exactly in terms of the degree-Kirchhoff index of the original graph. This limit corresponds to the thermodynamic limit in statistical mechanics and is fundamental in the physics of topological polymers. Furthermore, these asymptotic contraction factors are shown to fit well with molecular dynamics simulations, which should be useful for predicting the $g$-factors of topological polymer models with excluded volume.
\end{abstract}
\date{\today}


\maketitle

\section{Introduction}
We consider a classical model of polymers, discussed by James, Guth, and Flory~\cite{James:1943hc,James:1947hp,FloryPaulJ1969Smoc} and called \emph{phantom network theory}, in which the polymer molecule consists of a collection of monomers connected by displacement vectors representing effective bond vectors between adjacent monomers. Here each bond vector represents a Kuhn length (or a multiple of the Kuhn length) along the polymer, as in~\cite{rubenstein:colby}. For a linear polymer, we may think of the displacement vectors as independently sampled from multivariate Gaussian distributions, yielding a Gaussian random walk. The effective bond vectors (or displacement vectors) in a ring polymer must obey the additional condition that they must sum to zero, meaning that they are not independently sampled. Recently, polymers with more complicated topologies have been synthesized~\cite{Suzuki:2014fo,Tezuka:2017gh}, leading to an interest in modeling \emph{topological} polymers where the underlying structure is not a path or a cycle but an arbitrary connected multigraph $\graphG$. This introduces a more complicated dependence structure between displacement vectors.

To describe the model in this case, it's helpful to introduce some notation:
\begin{definition} 
Let $\graphG$ be an arbitrary connected multigraph (loop edges and multiple edges are allowed) with an orientation on each edge. A \emph{vertex vector} for $\graphG$ is an $x \in (\mathbb{R}^{d})^\verticesV$ where $x_i \in \mathbb{R}^d$ is the position of vertex $\vertex_i$ and $x^k \in \mathbb{R}^\verticesV$ is the vector of $k$-th coordinates of all vertex positions. An \emph{edge vector} $w$ for $\graphG$ is a $w \in (\mathbb{R}^d)^\edgesE$ where $w_j \in \mathbb{R}^d$ is the displacement along edge $\edge_j$ and $w^k \in \mathbb{R}^\edgesE$ is the vector of all $k$-th coordinates of the edge displacements. These are illustrated in Figure~\ref{fig:edge and vertex vectors}.	
\end{definition}
\begin{figure}
	\raisebox{-0.5\height}{\includegraphics[width=1.2in]{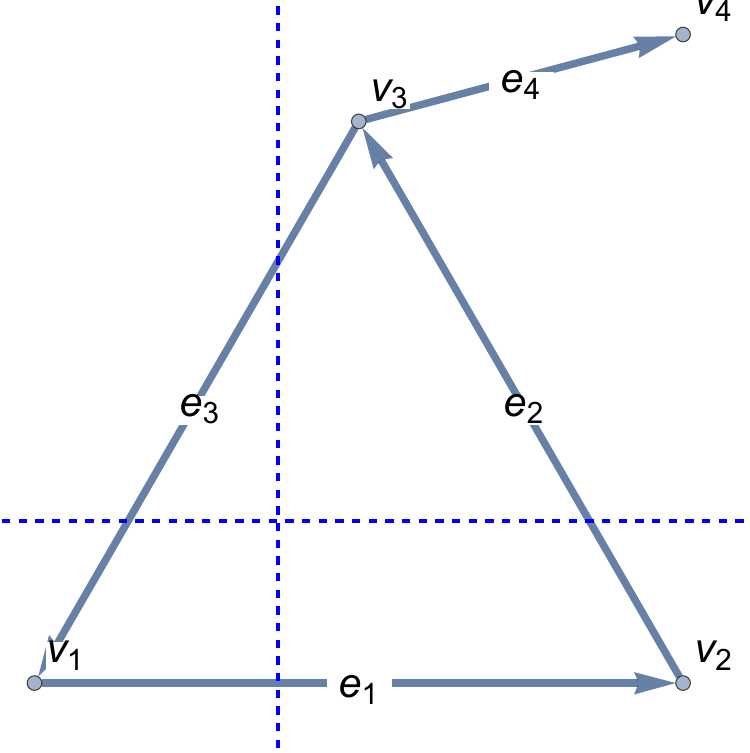}} \hspace{0.1in} 
	\begin{tabular}{c}
	  $x = \left( \right.$ \raisebox{-0.35\height}{\includegraphics[width=2.3in]{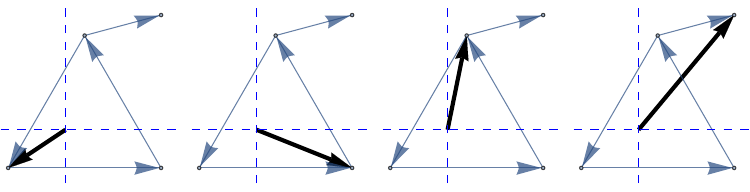}} $\left.\right)$ \\
	  $w = \left( \right.$ \raisebox{-0.35\height}{\includegraphics[width=2.3in]{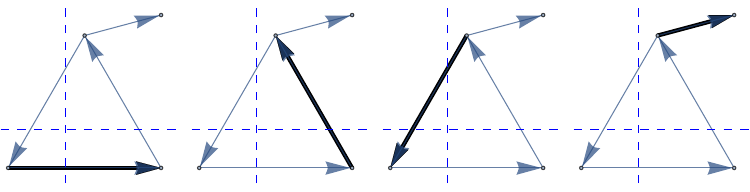}} $\left.\right)$ 
	  \end{tabular}
\caption{A particular graph embedding in $\R^2$, along with the components of its vertex vector $x \in (\R^{2})^4$ and edge vector $w \in (\R^{2})^4$.}
\label{fig:edge and vertex vectors} 
\end{figure}

The vertex and edge vectors are related by the $\verticesV \times \edgesE$ \emph{incidence matrix} $B$ of the graph $\graphG$, where 
\begin{equation*}
B_{ij} = \begin{cases}
			+1, & \text{if $\vertex_i$ is (only) the head of $\edge_j$} \\
			-1, & \text{if $\vertex_i$ is (only) the tail of $\edge_j$} \\
			 0, & \text{if $\vertex_i$ is both the head and tail of $\edge_j$}\\
			 0, & \text{if $\vertex_i$ is neither the head nor tail of $\edge_j$}
		 \end{cases}
\end{equation*}
and for each coordinate $k$, $w^k = B^T x^k$. We note that multiple $\edge_j$ may share the same head and tail vertices; these are simply repeated columns in $B$. While any $x^k \in \R^\verticesV$ may be a vertex vector, only $w^k \in \im B^T$ can be an edge vector. These $w^k$ have the special property that the sum of $w^k_i$ around any loop in $\graphG$ vanishes. We call $w$ an \emph{embeddable} edge vector if every $w^k \in \im B^T$. We can then define	

\begin{definition}
A \emph{Gaussian topological polymer} or \emph{topologically constrained random walk} (TCRW) in $\R^d$ with underlying (multi)graph $\graphG$ has all $w^k$ sampled independently from a standard normal distribution on the embeddable edge vectors $\im B^T \subset \R^{\edgesE}$.
\end{definition} 

We have previously shown~\cite{tcrw,MegaTCRW} that Gaussian TCRWs are exactly Gaussian phantom networks. We also proved that
\begin{theorem}[{\cite{tcrw}}]
Let $A^+$ denote the Moore--Penrose pseudoinverse of a matrix $A$, and $L(\graphG) = BB^T$ denote the graph Laplacian of a connected multigraph $\graphG$. If $x$ is a vertex vector for the TCRW with underlying graph $\graphG$, then the $x^k \in \R^{\verticesV}$ are sampled independently from multivariate normals with mean 0 and covariance matrix $L(\graphG)^+$.
\label{thm:vertex distribution}
\end{theorem}

In this paper, we will study the distribution of squared distances $\norm{x_i - x_j}^2$ between vertices in Gaussian TCRWs, and hence in phantom networks. We will first note that the expectation $\mathcal{E}(\norm{x_i - x_j}^2;\graphG)$ is equal to the resistance distance~\cite{Klein:1993bg,Gutman:1996hq,Mohar:1993gh} between $\vertex_i$ and $\vertex_j$ in $\graphG$ and show that the expectation $\mathcal{E}(R_g^2;\graphG)$ of the squared radius of gyration of a Gaussian TCRW is given in terms of the Kirchhoff index (or quasi-Wiener index) of $\graphG$. 

In practice, most of the graphs used to model topological polymers are constructed from a multigraph $\graphG$ by replacing each edge in $\graphG$ by a chain of $n$ edges for some large $n$. Calling the resulting graph $\graphG_n$, we are interested in determining the asymptotics of the Kirchhoff index of $\graphG_n$ as $n\to \infty$. 

\begin{theorem}
If $\graphG$ is a connected multigraph and $\graphG_n$ is the graph obtained by subdividing each edge of $\graphG$ into $n$ pieces, then 
\begin{equation*}
\lim_{n \rightarrow \infty} \frac{1}{n^3} \Kirchhoff(\graphG_n) = \frac{2 \Loops(\graphG) - 1}{12} \edgesE(\graphG) + \frac{1}{4} \DegreeKirchhoff(\graphG)
\end{equation*}
where $\Loops(\graphG) = \edgesE(\graphG) - \verticesV(\graphG) + 1$ is the cycle rank of $\graphG$ and $\DegreeKirchhoff(\graphG)$ is the ``degree-Kirchhoff index'' introduced by Chen and Zhang~\cite{Chen:2007dv}:
\begin{equation*}
\DegreeKirchhoff(\graphG) = \sum_{\vertex_i < \vertex_j \in \graphG} (\deg_\graphG \vertex_i \deg_\graphG \vertex_j ) r_{\vertex_i \vertex_j}.
\end{equation*}
\label{thm:degree Kirchhoff formula}
\end{theorem}

Along the way, we will prove some independently interesting results about resistances measured between points along edges in a resistor network.

Given the connection between Kirchhoff index and expected radius of gyration, this tells us something about the asymptotics of radius of gyration as we subdivide edges of a graph, corresponding to the thermodynamic limit in statistical mechanics as the system size goes to infinity. This result is most naturally expressed in terms of the \emph{contraction factor} or $g$-factor of a polymer, which is defined to be the ratio of the (expected) radius of gyration of the polymer to the expected radius of gyration of a linear polymer of the same length embedded in a space of the same dimension.

\begin{theorem}
For any connected multigraph $\graphG$ (including loop and multiple edges), if $\graphG_n$ is the graph created by dividing each edge of $\graphG$ into $n$ pieces, then the contraction factor obeys
\begin{equation*}
\lim_{n \rightarrow \infty} g(\graphG_n) = 
\frac{3}{\edgesE(\graphG)^2} \left( \operatorname{tr} \mathcal{L}^+(\graphG) + \frac{1}{3} \operatorname{Loops}(\graphG) - \frac{1}{6} \right).
\end{equation*}
Here $g(\graphG_n)$ is the $g$-factor of $\graphG_n$, $\operatorname{Loops}(\graphG) = \edgesE(\graphG) - \verticesV(\graphG) + 1$ is the cycle rank of $\graphG$, $\mathcal{L}(\graphG)$ is the \emph{normalized} graph Laplacian of $\graphG$, and $\mathcal{L}^+(\graphG)$ is the Moore--Penrose pseudoinverse of $\mathcal{L}(\graphG)$.
\label{thm:contraction factors}
\end{theorem}

We will use this theorem to analyze the relative sizes of the topological polymers in solution synthesized by Tezuka and compare the scaling coefficients with a molecular dynamics calculation of expected radii of gyration. We will see that there is an excellent linear fit between our coefficients and the results of the simulation. We suggest that the linear fit should be useful for estimating the $g$-factor of real topological polymers from the calculations with the ideal chain model. Recall that the estimates of the $g$-factor evaluated in molecular dynamics are for real topological polymer models with excluded volume, and hence are quite expensive to compute, while our coefficients are simple to compute.

The asymptotic contraction factor should be one of the most fundamental physical or dynamical quantities for topological polymers in solution. In fact, one can use this quantity to estimate various physical quantities of topological polymers in solution such as the viscosity coefficient. The asymptotic contraction factor should also be relevant to experiments on the viscosity of polymer solutions~\cite{rubenstein:colby}.   

The asymptotic contraction factors themselves could be useful for studying the mean-square radius of gyration for various real topological polymer models with excluded volume~\cite{zhu_radius_2016,uehara_statistical_2016}. It has been shown that the ratios among the estimates of the mean-square radius of gyration for real topological polymer models with different graphs and those of ideal topological polymer models with the same set of graphs are almost the same if the functionality (i.e., the degree) at each vertex is limited up to three~\cite{uehara_statistical_2016}. Thus, even the asymptotic value of the mean-square radius of gyration for an ideal topological polymer model with graph $\graphG_n$ can be useful to estimate that of a real topological polymer model with the corresponding graph.    

There is another interesting approach to the spectrum of the graph Laplacian for a topological polymer whose edges are subdivided into $n$ pieces.  A method for reducing the graph Laplacian of the Gaussian topological polymer of graph $\graphG_n$ which is obtained by replacing each chain of a given graph $\graphG$ by a chain of $n$ edges was derived~\cite{eichinger_distribution_1978}. In order to evaluate the $g$-factor by the method, one has to evaluate all the eigenvalues of the reduced matrix, which is not practical, in general. However, some information on the spectrum of eigenvalues can be investigated through it.

\section{The multigraph Laplacian and resistor networks}

We commented above that the matrix $L = BB^T$ was known as the graph Laplacian of $\graphG$. We now expand on this point, summarizing some of the widely developed and rich theory of graph Laplacians. First, we observe that $L$ is well-defined for graphs $\graphG$ with multiple edges and loop edges because $B$ is defined for such graphs. For multigraphs, the \emph{degree} $\deg \vertex_i$ of a vertex i counts the number of edges leaving or arriving at $\vertex_i$ (loop edges count twice), the degree matrix $D$ is the diagonal matrix of vertex degrees and the adjacency matrix $A$ is defined by $-A_{ij} = \# \text{ one edge paths from $\vertex_i$ to $\vertex_j$}$. Every loop edge $\vertex_i \rightarrow \vertex_i$ provides two different paths from $\vertex_i$ to $\vertex_i$, and so contributes $2$ to the diagonal matrix. 

An easy computation shows $L = D - A$, or 
\begin{equation*}
L_{ij} = \begin{cases}
			\deg \vertex_i - 2 (\text{\# edges $\vertex_i \rightarrow \vertex_i$ in $\graphG$})
			& \text{if $i = j$}, \\
			-\text{(\# edges $(\vertex_j \rightarrow \vertex_i$ or $\vertex_i \rightarrow \vertex_j) \in \graphG$)} 
			& \text{if $i \neq j$}. 
		\end{cases}
\end{equation*}
For functions $p \co\!\! \verticesV(\graphG) \rightarrow \R$ defined on the vertices of a multigraph $\graphG$ viewed as vectors in $\R^\verticesV$ we can interpret the Laplacian as a linear operator. Explicitly
\begin{equation}
L p (\vertex_i) = (\deg \vertex_i) \, p(\vertex_i) - \hspace{-0.2in} \sum_{\vertex_j \text{ adjacent to } \vertex_i} \hspace{-0.2in} p(\vertex_j) 
\label{eq:local laplacian}
\end{equation}
where $\vertex_j$ is \emph{adjacent} to $\vertex_i$ in $\graphG$ if there is an edge $\vertex_i \rightarrow \vertex_j$ or $\vertex_j \rightarrow \vertex_i$ in $\graphG$.
Note that for each loop edge $\vertex_i \rightarrow \vertex_i$ we increment $\deg \vertex_i$ by 2 \emph{and} count $\vertex_i$ as adjacent to itself in two ways. These contributions cancel in~\eqref{eq:local laplacian}.

For functions $p \co\!\! \verticesV(\graphG) \rightarrow \R$ defined on the vertices of a graph $\graphG$ viewed as vectors in $\R^\verticesV$ we can interpret $B^T$ as the \emph{gradient} operator. If $\edge_k = \vertex_i \rightarrow \vertex_j$,
\begin{equation*}
\grad p(\edge_k) = p(\vertex_j) - p(\vertex_i) = (B^T p)_k.
\end{equation*}
For functions $u \co\!\! \edgesE(\graphG) \rightarrow \R$ defined on the edges of $\graphG$ and viewed as vectors in $\R^\edgesE$, we can interpret $B$ as the negative of the \emph{divergence} operator
\begin{equation*}
\div u(\vertex_i) = \sum_{\edge_j = \vertex_i \rightarrow \cdot} u(\edge_j) - \sum_{\edge_j = \cdot \rightarrow \vertex_i} u(\edge_j) = (-Bu)_i.
\end{equation*}
This means that the graph Laplacian $L = -\grad \div$. This agrees with the sign convention in Riemannian geometry~\cite{ChavelIsaac.1984EiRg}, but is opposite from the sign convention in mathematical physics~\cite{1962Momp}.

If the graph $\graphG$ is a conductive network, then a function $p\co\!\! \verticesV \rightarrow \R$ is called a \emph{potential}. Along every edge of $\graphG$, the difference in potential $\grad p = U$ is the corresponding \emph{voltage}.  Ohm's law states that $U = IR$ where $I$ is the current and $R$ the resistance, so if each edge has unit resistance, $I = U$. Kirchhoff's first law says that the total current flowing into and out of each vertex are equal, or that $\div I = 0$. Combining these, we see that $Lp = 0$ in an isolated system. 

If $\graphG$ is connected, the kernel of $L$ is known to be one-dimensional and spanned by the constant functions. This means that in an isolated system, $p$ is a constant and no current flows. However, if we introduce unit current at $\vertex_i$ and remove unit current at $\vertex_j$, then the potential is the minimum-norm solution $p^{ij}$ to the Poisson problem
\begin{equation}
(Lp^{ij})(\vertex_k) = e_i - e_j = \begin{cases} 
				             	+1, & \text{if $k = i$}, \\
								-1, & \text{if $k = j$}.
							  \end{cases}
\label{eq:poisson problem 1}
\end{equation}
Note that the right hand side is in $(\ker L)^\perp = \im L^T = \im L$, so $p^{ij}$ is an exact solution and not a least-squares solution. The resistance $r_{ij}$ between $\vertex_i$ and $\vertex_j$ is then 
\begin{equation}
r_{ij} := p^{ij}(\vertex_i) - p^{ij}(\vertex_j).
\label{eq:resistance distance definition}
\end{equation}
In graph theory~\cite{Klein:1993bg}, this value is called the \emph{resistance distance} from $\vertex_i$ to $\vertex_j$ in $\graphG$. 

To analyze the resistance distance, we will need to recall the theory of the Moore--Penrose pseudoinverse of a matrix $A$. First, if $A$ has singular value decomposition $A = U \Sigma V^T$ with singular values $\sigma_i$ then $A^+ := V \Sigma^+ U^T$, where $\Sigma^+$ is a diagonal matrix with entries
\begin{equation*}
\sigma'_i = \begin{cases} 
				\frac{1}{\sigma_i}, & \text{if $\sigma_i \neq 0$},\\
				0, &\text{if $\sigma_i = 0$}.
			\end{cases}
\end{equation*}
The pseudoinverse has many useful properties, but the one that we will use the most often is the fact that $x = A^+ b$ is the unique minimum-norm least-squares solution to the linear system $A x = b$. Explicitly, this means that $x \in \ker A^\perp$ and $Ax$ is the orthogonal projection of $b$ to $\im A$. It is also true that $A^+ A$ is orthogonal projection to $\ker A^\perp$ and $A A^+$ is orthogonal projection to $\im A$. 

We note that this property of the pseudoinverse implies that minimum-norm least-squares solutions to linear systems are linear in their right-hand sides. If $x$ and $y$ solve $A x = b$ and $A y = c$ in this sense, then $x = A^+b$ and $y = A^+c$. But if $w$ is the minimum-norm least squares solution to $A w = b+c$, then $w \in (\ker A)^{\perp}$ so applying $A^+$ to both sides we get 
\begin{equation}
w = A^+ A w = A^+(b+c) = A^+b + A^+c = x + y.
\label{eq:superposition of solutions}
\end{equation}
In differential equations, this property is usually called ``superposition of solutions''.

We now define the potential function $p^i$ to be the unique minimum-norm least-squares solution to the Poisson problem $L p^i = e_i$. Since $e_i$ is not in $\im L$, the solution $p^i$ exactly solves the slightly different Poisson problem
\begin{equation}
L p^i(\vertexw) = \begin{cases}
					1-\frac{1}{\verticesV}, & \text{if $\vertexw = \vertex_i$} \\
					-\frac{1}{\verticesV}, & \text{otherwise.}
				   \end{cases}
\label{eq:poisson problem 1a}
\end{equation}

We can now prove
\begin{proposition}
If $L$ is the graph Laplacian of a connected multigraph $\graphG$, the resistance distance $r_{ij}$ between vertices $v_i$ and $v_j$, the potential functions $p^{ij}$, $p^i$, and $p^j$ and the matrix $L^+$ are related in the following ways:
\begin{equation*}
p^{ij} = L^+(e_i - e_j), \quad p^{i} = L^+ e_i, \quad p^{j} = L^+ e_j, \quad p^{ij} = p^{i} - p^{j},
\end{equation*}
while
\begin{align}
r_{ij} = p^{ij}(\vertex_i) - p^{ij}(\vertex_j) &= p^{i}(v_i) - p^{j}(v_i) - p^i(v_j) + p^j(v_j) \nonumber \\
       &= (L^+)_{ii} - (L^+)_{ij} - (L^+)_{ji} + (L^+)_{jj}.
\label{eq:resistance distance formula}
\end{align}
Further, $L^+ L p^{ij} = p^{ij}$ while $L^+L p^i = p^i$.
\label{prop:resistance distance formula}
\end{proposition}
\begin{proof}
Since $p^{ij}$ and $p^{i}$ are defined to be minimum-norm least squares solutions to Poisson problems, they are given by $p^{ij} = L^+ (e_i - e_j)$ and $p^i = L^+ e_i$. The linearity property~\eqref{eq:superposition of solutions} then implies that $p^{ij} = p^i - p^j$. Then~\eqref{eq:resistance distance formula} follows from~\eqref{eq:resistance distance definition} since $p^i = L^+ e_i$ is the $i$th row of $L^+$, so its values at $\vertex_i$ and $\vertex_j$ are the matrix elements $(L^+)_{ii}$ and $(L^+)_{ij}$. Finally, since $p^{ij}$ and $p^i$ are \emph{minimum-norm} solutions to Poisson problems, they are both in $(\ker L)^\perp$ and hence fixed by $L^+L$ (the orthogonal projection onto $(\ker L)^\perp$).
\end{proof}
The sum of all the resistance distances between vertices in a graph is known as the Kirchhoff index (or quasi-Wiener index) of the graph and denoted 
\begin{equation*}
\Kirchhoff(\graphG) = \sum_{i < j} r_{ij}.
\end{equation*}
The Kirchhoff index is a measure of the connectivity of the graph-- graphs which are very well connected with short paths between vertices have small Kirchhoff indices while graphs with longer paths between vertices have larger ones. There is a well-developed theory of the Kirchhoff index in mathematical chemistry~\cite{Klein:1993bg,Klein:2002vx,bonchev_molecular_1994,Gutman:1996hq,klein_random_2004,Bapat:cx,Zhang:2006bm,babic_resistance-distance_2002}.

\section{Expected squared chordlengths and resistances}

From Theorem~\ref{thm:vertex distribution}, it follows immediately that the marginal distribution of the displacement vector $x_i - x_j$ between vertices $\vertex_i$ and $\vertex_j$ is Gaussian with mean zero. The variance of the Gaussian can be computed by taking
\begin{equation*}
\mathcal{E}\left((x_i^k - x_j^k)^2;\graphG\right) = \mathcal{E}\left( (x_i^k)^2 + (x_j^k)^2 - 2 (x_i^k x_j^k);\graphG \right) = L^+_{ii} + L^+_{jj} - L^+_{ij} - L^+_{ji}
\end{equation*}
since $L^+$ is the covariance matrix of the $x_i^k$. This proves
\begin{proposition}
Let $\vertex_i$ and $\vertex_j$ be vertices in a multigraph $\graphG$ and $r_{ij}$ be the resistance distance between them. The expected squared distance between $\vertex_i$ and $\vertex_j$ in a Gaussian TCRW with multigraph $\graphG$ in $\R^d$ is given by $\mathcal{E}\left(\norm{x_i - x_j}^2;\graphG \right) = d (L^+_{ii} + L^+_{jj} - L^+_{ij} - L^+_{ji}) = d \, r_{ij}$.
\label{prop:expected squared distance Lplus}
\end{proposition}

We now give two examples. It is well known that in a Gaussian random polygon (ring polymer) in $\R^3$ with $n$ vertices, the expected squared chordlength between vertices $\vertex_1$ and $\vertex_j$ separated by $j$ edges is given by $3j(n-j)/n$~\cite{1949JChPh..17.1301Z}. We may recompute this result rather simply using resistances. 

There are two paths from $\vertex_1$ to $\vertex_j$, one composed of $j$ resistors in series and the other composed of $n-j$ resistors in series. The two paths are in parallel, so the total resistance is
\begin{equation}
r_{1j} = \frac{1}{\frac{1}{j} + \frac{1}{n-j}} = \frac{j(n-j)}{n}.
\label{eq:expected squared distance in cycle graph}
\end{equation}

Extending this idea, suppose we have a ``multitheta'' graph consisting of $m$ arcs of $n$ edges joining two vertices. The total resistance of each arc is $n$, and $m$ such arcs are in parallel, so the resistance between junctions is 
\begin{equation*}
\frac{1}{\frac{1}{n} + \cdots + \frac{1}{n}} = \frac{1}{\frac{m}{n}} = \frac{n}{m}.
\end{equation*}
This recovers an asymptotic result of Deguchi and Uehara~\cite{Uehara:2018bb} and, independently, of Zhu, Wang, Li, and Wang~\cite{zhu_radius_2016}. It also shows the interesting phenomenon that the junctions are expected to be closer together when more arcs join them, even if the length of the arcs remains constant. 

\section{The radius of gyration}

A standard measure of the effective size of a polymer in solution is the radius of gyration.
\begin{definition}
The radius of gyration of a TCRW with vertex vector $x$ in $\R^d$ is given by
\begin{equation*}
R_g^2(x) = \frac{1}{2 \verticesV^2} \sum_{i,j \in 1}^\verticesV \norm{x_i - x_j}^2.
\end{equation*}
The expected radius of gyration of a Gaussian TCRW with underlying graph $\graphG$ is given  by the expectation of $R_g^2(x)$ when each coordinate vector $x^k$ is chosen according to a multivariate Gaussian with mean zero and covariance matrix $L^+(\graphG)$. In polymer science, this is the mean-square radius of gyration~$\left< s^2 \right>$ of the molecule~(cf. \cite[1.17]{stepto:2015be}); we denote it by $\mathcal{E}(R_g^2;\graphG)$. 
\label{def:radius of gyration}
\end{definition}
We can use our result connecting resistance distances with expectations of chordlengths to compute the expected radius of gyration rather simply, recovering a classical formula for the radius of gyration of a molecule of arbitrary topology~\cite[eq.\ 18a]{Eichinger1980}:
\begin{theorem}
For any Gaussian TCRW in $\R^d$ with multigraph $\graphG$, we have 
\begin{equation*}
\mathcal{E}\!\left(R_g^2;\graphG\right) = \frac{d}{\verticesV^2} \Kirchhoff(\graphG) = \frac{d}{\verticesV} \tr L^+ = \frac{d}{\verticesV} \sum_{i=1}^{\verticesV - 1} \frac{1}{\lambda_i}
\end{equation*}
where $\lambda_i$ are the nonzero eigenvalues of the graph Laplacian $L(\graphG)$.
\label{thm:expected radius of gyration}
\end{theorem}
\begin{proof}
Substituting the result of~Proposition~\ref{prop:expected squared distance Lplus} into Definition~\ref{def:radius of gyration}, we get 
\begin{equation*}
R_g^2 = \frac{d}{2 \verticesV^2} \sum_{i,j} (L^+)_{ii} + (L^+)_{jj} - 2 (L^+)_{ij} 
= \frac{d}{2 \verticesV^2} \left( 2 \verticesV \tr L^+ - 2 \sum_{i,j} (L^+_{ij}) \right).
\end{equation*}
To analyze the last term on the right, we note that $\ker L = \ker L^+$, so in particular the constant vector of $1$'s is in $\ker L^+$. This means that the row sums of $L^+$ are zero, and hence that the ``grand sum'' $\sum_{i,j} (L^+)_{ij}$ is zero as well. The statement on eigenvalues is simply that the (nonzero) eigenvalues of $L^+$ are the reciprocals of those of $L$. 
\end{proof}
We note that the connection between Kirchhoff index and the trace of $L^+$ is well-known~\cite{Gutman:1996hq}, though it is not always stated clearly whether the result is intended to apply to multigraphs $\graphG$. 

We now give two examples for Theorem~\ref{thm:expected radius of gyration}. First, consider the path graph with $\verticesV$ vertices $\vertex_1,\dotsc,\vertex_{\verticesV}$. The resistance between vertex $\vertex_i$ and $\vertex_j$ is simply $\abs{i-j}$. Summing, we get
\begin{equation}\label{eq:path gyradius}
\mathcal{E}(R_g^2;\text{path graph}) = \frac{d}{6} \frac{\verticesV(\verticesV+1)}{\verticesV-1}
\end{equation}
which agrees with the standard asymptotic expression of $\frac{d \verticesV}{6}$.

The Gaussian ring polymer is based on the cycle graph with $\verticesV$ vertices. The graph Laplacian for this graph can be written
\begin{equation*}
L = 
\begin{pmatrix}
 2 & -1 &  0 &  0 & 0 & \dots & 1 \\
-1 &  2 & -1 &  0 & 0 & \dots & 0 \\
 0 & -1 &  2 & -1 & 0 & \dots & 0 \\
 \vdots & \vdots & \vdots & \vdots & \vdots \\
-1 &  0 &  0 & 0 & 0 & \dots & 2
\end{pmatrix}
\end{equation*}
which is a circulant matrix (every row is a cyclic shift of the previous one). The eigenvalues of a circulant matrix are easy to compute; in this case they are 
\begin{equation*}
\lambda_j = 2 - 2 \cos \left( \frac{2\pi j}{\verticesV} \right) = 4 \sin^2 \left(  \frac{\pi j}{\verticesV} \right)
\end{equation*}
for $j \in 0, \dots, \verticesV-1$. Only $\lambda_0 = 0$, so we have to sum $\sum_{j=1}^{\verticesV-1} \csc^2 (\frac{\pi j}{\verticesV} ) = \frac{1}{3}(\verticesV^2 - 1)$~\cite[4.4.6.5, p.\ 644]{Prudnikov:1986vp}. Zhang et al.~\cite{Zhang:2006bm} compute this along with formulae for the Kirchhoff index of other circulant graphs. This means that the expected radius of gyration is (exactly)
\begin{equation*}
\mathcal{E}(R_g^2;\text{cycle graph}) = \frac{d}{12} \frac{\verticesV^2 - 1}{\verticesV},
\end{equation*}
which agrees with both the standard asymptotic approximation $\frac{d \verticesV}{12}$ \cite{Casassa:1965hu,Kramers:1946ki} and the result of summing~\eqref{eq:expected squared distance in cycle graph} over all pairs of vertices in the cycle and dividing by $\verticesV^2$.

For more general graphs, there are a number of useful ways to approach the numerical computation of the Kirchhoff index. Finding the eigenvalues of $L$ is certainly the most straightforward. The fastest one in practice seems to be to calculate in terms of the last two coefficients of the characteristic polynomial of the graph Laplacian $L$, which is a sparse matrix.

\begin{proposition}{\cite{Mohar:1993gh}}
If $p(t) = \det(L - tI) = t^\verticesV + c_1 t^{\verticesV-1} + \dots + c_{\verticesV-1} t + c_{\verticesV}$ is the characteristic polynomial of the graph Laplacian $L(\graphG)$, then $\mathcal{E}(R_g^2;\graphG) = -\frac{1}{\verticesV}\frac{c_{\verticesV-2}}{c_{\verticesV-1}}.$
\label{prop:characteristic polynomial formula}
\end{proposition}

%\begin{proof}
%The coefficients of the characteristic polynomial are given by elementary symmetric functions of the $\lambda_i$,
%\begin{equation*}
%c_i = (-1)^i \sum_{\substack{\{a_1,\dots,a_i\}\\a_j \text{ all different}\\a_j \in \{1,\dots,\verticesV\}}} \lambda_{a1} \cdots \lambda_{ai}
%\end{equation*}
%In particular, because one of the eigenvalues is zero, $c_{\verticesV} = 0$, $c_{\verticesV-1} = (-1)^{\verticesV-1} \lambda_2 \cdots \lambda_{\verticesV}$ and 
%\begin{equation*}
%c_{\verticesV-2} = (-1)^{\verticesV-2} \sum_{j=2}^\verticesV \lambda_2 \cdots \hat{\lambda_j} \cdots \lambda_{\verticesV}
%\end{equation*}
%where $\hat{\lambda_j}$ means that $\lambda_j$ is missing from the product. Taking the quotient, we see that 
%\begin{equation*}
%-\frac{c_{\verticesV-2}}{c_{\verticesV-1}} = \sum_{j=2}^\verticesV \frac{1}{\lambda_j},
%\end{equation*}
%as desired.
%\end{proof}

We note that Proposition~\ref{prop:characteristic polynomial formula} has the perhaps surprising corollary:
\begin{corollary}
The expected radius of gyration $\mathcal{E}(R_g^2)$ of any Gaussian TCRW is a rational number.
\end{corollary}

\section{Subdivision graphs and topological polymers}

The topological polymers so far synthesized in the laboratory have long chains of monomers joining a relatively small number of ``junction'' molecules. Thus, while the number of edges in their graph may be large, the relative complexity of the graphs is rather modest. Further, since each edge represents a persistence length of the chain of monomers, it is not always clear exactly how many edges should be in each path between junctions. 

Therefore, we now study the behavior of the expected radius of gyration for a particular kind of subdivision graph as the number of subdivisions increases. Two cautions are in order: the usual definition of a subdivision graph places one vertex in the middle of each edge, so iterated subdivisions divide the original edges into a power of two subdivisions~\cite{Yang:2014im}. This is different from our model, where any number of subdivisions are allowed. Second, while previous authors subdivided simple graphs, we explicitly allow loop edges and multiple edges.
\begin{definition}
The $n$-part edge subdivision $\graphG_n$ of a multigraph $\graphG$ is the graph obtained by dividing each edge of $\graphG$ into $n$ pieces. 
\end{definition}
Figure~\ref{fig:subdivision} shows an example of the $5$-fold edge subdivision of a multigraph.
\begin{figure}
\hfill
\begin{overpic}[width=2.5in]{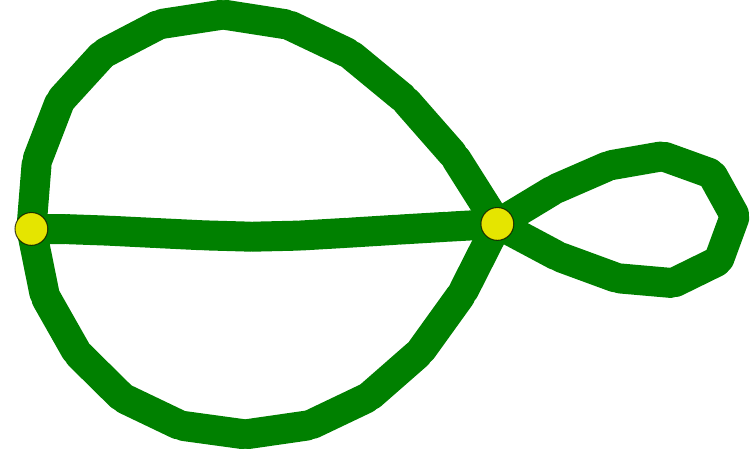}
\put(30,-5){$\graphG$}
\end{overpic}
\hfill
\begin{overpic}[width=2.5in]{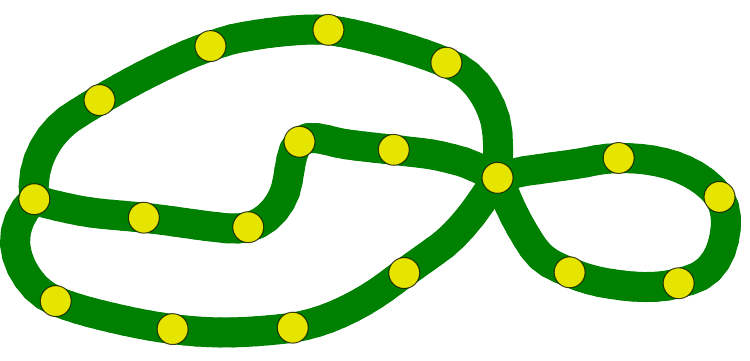}
\put(30,-5){$\graphG_5$}
\end{overpic}
\hfill
\hphantom{.}
\caption{The $n$-part edge subdivision divides each edge of $\graphG$ into $n$ pieces. On the left, we see a multigraph and on the right its $5$-part edge subdivision.}
\label{fig:subdivision}
\end{figure}

We have already analyzed the expected radius of gyration of two $n$-part edge subdivisions of a graph-- if $\graphG$ has two vertices and one edge joining them, the path graph with $n$ edges is $\graphG_n$ and 
\begin{equation}
\mathcal{E}\left( R_g^2 ; \graphG_n  = \text{ path graph}\right) = \frac{d}{6} \frac{(n+1)(n+2)}{n}.
\label{eq:path graph subdivision}
\end{equation}
If $\graphG$ has one vertex and one loop edge, the cycle graph with $n$ edges is $\graphG_n$. We saw that
\begin{equation}
\mathcal{E}\left( R_g^2 ; \graphG_n = \text{ cycle graph}\right) = \frac{d}{12} \frac{n^2 - 1}{n}.
\label{eq:cycle graph subdivision}
\end{equation}
Notice that both of these are asymptotically linear in $n$. We will show that this is a general feature of edge-subdivision graphs, and show how to compute the leading coefficient. 

To do so, we need to recall a definition:
\begin{definition}
The \emph{normalized} graph Laplacian $\mathcal{L}(\graphG)$ is given by
\begin{equation*}
	\mathcal{L}_{ij} = \begin{cases}
					1 - \frac{2\times\text{\# loop edges}}{\deg(\vertex_i)}, & \text{ if } i = j, \\
					-\frac{k}{\sqrt{\deg(\vertex_i) \deg(\vertex_j)}}, &\text{ if } \vertex_i, \vertex_j \text{ joined by $k$ edges},\\
					0, & \text{ otherwise.}
				\end{cases}
\end{equation*}	
\end{definition}

We can now state our next main result on the asymptotics of radius of gyration:
\begin{theorem}
For any connected multigraph $\graphG$ (including loop and multiple edges), if $\graphG_n$ is the $n$-part edge subdivision of $\graphG$, then the expected radius of gyration of a Gaussian TCRW with underlying graph $\graphG_n$ embedded in $\R^d$ obeys
\begin{equation*}
\lim_{n \rightarrow \infty} \frac{1}{n} \mathcal{E}(R_g^2; \graphG_n) = 
\frac{d}{2 \edgesE(\graphG)} \left( \operatorname{Tr} \mathcal{L}^+(\graphG) + \frac{1}{3} \operatorname{Loops}(\graphG) - \frac{1}{6} \right).
\end{equation*}
Here $\operatorname{Loops}(\graphG) = \edgesE(\graphG) - \verticesV(\graphG) + 1$ is the cycle rank of $\graphG$.
\label{thm:main}
\end{theorem}
For the graph $\graphG$ with two vertices and one edge, $\mathcal{L}$ is the $2\times 2$ matrix $\left( \begin{smallmatrix} 1 & -1 \\ -1 & 1 \end{smallmatrix} \right)$. A computation reveals that $\mathcal{L}^+ = \frac{1}{4} \mathcal{L}$ and $\tr \mathcal{L}^+ = \frac{1}{2}$. Thus Theorem~\ref{thm:main} recovers the coefficient $\frac{d}{6}$ from~\eqref{eq:path graph subdivision}. Similarly, for the graph $\graphG$ with one vertex and one loop edge, $\mathcal{L}$ is the $1 \times 1$ matrix whose single entry is zero, so $\tr \mathcal{L}^+ = 0$. Thus Theorem~\ref{thm:main} also recovers the coefficient $\frac{d}{12}$ from~\eqref{eq:cycle graph subdivision}. 

We will now develop several tools leading up to the proof of Theorem~\ref{thm:main}. Our basic idea is to calculate resistance distances on the subdivided graph $\graphG_n$ by constructing solutions of the Poisson problem~\eqref{eq:poisson problem 1} using superposition of solutions. In this way, we'll be able to relate resistances on $\graphG_n$ to resistances on $\graphG$, and so compute the Kirchhoff index of $\graphG_n$ by summing these formulae. This approach was inspired by the Green's kernel method of Carmona, Mitjana, and Mons\'o~\cite{Carmona:2017ge}.

\subsection{Poisson problems on subdivided graphs}

For this section, we establish the notation that we will refer to vertices along a subdivided edge in $\graphG_n$ as $\vertex_0, \dotsc, \vertex_n$. The end vertices $\vertex_0$ and $\vertex_n$ correspond to vertices in $\graphG$ which we will refer to as $\vertexa$ and $\vertexb$. It is possible that $\vertexa = \vertexb$ if the original edge in $\graphG$ was a loop edge. 

We first note that functions on subdivided edges are linear when their Laplacian vanishes.
\begin{lemma}
If a vertex $\vertex_i$ is joined to two other vertices $\vertex_{i-1}$ and $\vertex_{i+1}$ then $Lp(\vertex_i) = 0$ if and only if $p$ is a linear function of $i$ on $\vertex_{i-1}, \vertex_i, \vertex_{i+1}$.
\label{lem:harmonic implies linear}
\end{lemma}

\begin{proof} 
Using~\eqref{eq:local laplacian}, the Laplacian at $\vertex_i$ is given by
\begin{equation*}
L p (\vertex_i) = -p(\vertex_{i-1}) + 2 p(\vertex_i) - p(\vertex_{i+1}).
\end{equation*}
Rearranging the right hand side, we see $L p (\vertex_i) = 0$ if and only if $p(\vertex_i) - p(\vertex_{i-1}) = p(\vertex_{i+1}) - p(\vertex_i)$, which is exactly the condition for $p$ to be a linear function on these points.
\end{proof}

We can now define and solve a particular Poisson problem on a subdivided edge
\begin{proposition}
Suppose $\vertexa \rightarrow \vertexb$ is an edge of a connected multigraph $\graphG$ subdivided in $\graphG_n$ by $\vertex_0, \dotsc, \vertex_n$. The Poisson problem  
\begin{equation}
L p^{\vertexa j \vertexb}(\vertexw) = 
\begin{cases}
	-\frac{n-j}{n}, & \text{if $\vertexw = \vertexa \neq \vertexb$} \\
	-\frac{j}{n}, & \text{if $\vertexw = \vertexb \neq \vertexa$} \\
	-1, & \text{if $\vertexw = \vertexa = \vertexb$}\\
	+1, & \text{if $\vertexw = \vertex_j$}\\
	0, & otherwise
\end{cases}
\label{eq:poisson problem 2}
\end{equation}
has minimum norm solution 
\begin{equation}
p^{\vertexa j \vertexb}(\vertexw) = C + 
\begin{cases}
  \frac{n-j}{n} k, & \text{if $\vertexw = \vertex_k$ and $k \leq j$} \\
  -\frac{j}{n} k + j, & \text{if $\vertexw = \vertex_k$ and $k \geq j$} \\
  0, & \text{otherwise, including $\vertexw = \vertexa$ and $\vertexw = \vertexb$} 
\end{cases}
\label{eq:pajb solution}
\end{equation}
for some $C$. Since $L p^{ajb}$ is in $\im L$, $L^+ L p^{ajb} = p^{ajb}$.
\label{prop:pajb}
\end{proposition}

\begin{proof}
The function $p^{\vertexa j \vertexb}$ is piecewise linear with corners at $\vertexa$, $\vertexb$ and $\vertex_j$. This proves that $L p^{\vertexa j \vertexb} = 0$ away from these three vertices by Lemma~\ref{lem:harmonic implies linear}. The Laplacian may be checked directly at those points using~\eqref{eq:local laplacian}. Since all solutions of Poisson problems on a connected graph $\graphG$ differ by a constant there must be some $C$ which yields the minimum norm solution. 
\end{proof}

Recalling that $p^{ij}$ is our notation for a solution to~\eqref{eq:poisson problem 1}, we can now prove a general decomposition formula for solutions of~\eqref{eq:poisson problem 1} on a subdivided graph $\graphG_n$.

\begin{proposition}
Suppose that $\vertexa \rightarrow \vertexb$ and $\vertexc \rightarrow \vertexd$ are any edges in $\graphG$ (they may be loop edges, connect the same endpoints, or even be the same edge). Suppose that $\vertexq$ is $j$ edges from $\vertexa$ along the subdivision of $\vertexa \rightarrow \vertexb$ in $\graphG_n$ and $\vertexr$ is $k$ edges from $\vertexc$ along the subdivision of $\vertexc \rightarrow \vertexd$ in $\graphG_n$. 

Then as functions on $\graphG_n$, we have the following relation between minimum norm least-squares solutions to the three Poisson problems \eqref{eq:poisson problem 1},~\eqref{eq:poisson problem 1a} and~\eqref{eq:poisson problem 2}:
\begin{equation*}
p^{\vertexq \vertexr} = p^{\vertexa j \vertexb} 
+ \frac{n-j}{n} p^{\vertexa} 
+ \frac{j}{n} p^{\vertexb} 
- \frac{n - k}{n} p^{\vertexc}
- \frac{k}{n} p^{\vertexd}
- p^{\vertexc k \vertexd}
\end{equation*}
where we use the convention $p^{\vertexa \vertexa} = 0$ and note that $p^{\vertexa 0 \vertexb} = p^{\vertexa n \vertexb} = 0$ in case there are coincidences among $\vertexa, \vertexb, \vertexc, \vertexd, \vertexq$ and $\vertexr$.  
\label{prop:decomposition of solutions to poisson problems}
\end{proposition}

\begin{proof}
We first observe that 
\begin{equation*}
L p^{\vertexq} =  L p^{\vertexa j \vertexb} + \frac{n-j}{n} L p^{\vertexa} + \frac{j}{n} L p^{\vertexb}.
\end{equation*}
Both sides are equal to $-\frac{1}{\verticesV}$ away from $\vertexq$, $\vertexa$, and $\vertexb$. At $\vertexq$, $\vertexa$, and $\vertexb$ we can use the definitions of the functions to check the equality directly. Equality still holds if $\vertexa = \vertexb$, $\vertexq = \vertexa$ (that is, $j=0$) or $\vertexq = \vertexb$ (that is, $j=n$). Applying $L^+$ to both sides of the equation, we have shown 
\begin{equation*}
p^{\vertexq} =  p^{\vertexa j \vertexb} + \frac{n-j}{n} p^{\vertexa} + \frac{j}{n} p^{\vertexb}.
\end{equation*}
regardless of any possible coincidences between $\vertexa$, $\vertexb$ and $\vertexq$. 

Exactly the same argument shows 
\begin{equation*}
p^{\vertexr}=  p^{\vertexc k \vertexd} + \frac{n-k}{n} p^{\vertexc} + \frac{k}{n} p^{\vertexd}
\end{equation*}
regardless of coincidences among $\vertexc$, $\vertexd$ and $\vertexr$. 

Since $p^{\vertexq \vertexr} = p^{\vertexq} - p^{\vertexr}$ this proves the result. We note that this last equality holds regardless of any coincidences between $\{ \vertexa, \vertexb, \vertexq \}$ and $\{ \vertexc, \vertexd, \vertexr \}$.
\end{proof}

\subsection{Relating resistances between middle points and endpoints of subdivided edges}

We now use the decomposition formula of Proposition~\ref{prop:decomposition of solutions to poisson problems} to compute the resistance between vertices in the middle of subdivided edges in terms of resistances between their endpoints. 

\begin{proposition}
Suppose that $\vertexa \rightarrow \vertexb$ and $\vertexc \rightarrow \vertexd$ are any edges in $\graphG$ (they may be loop edges, connect the same endpoints, or even be the same edge). Suppose that $\vertexq$ is $j$ edges from $\vertexa$ along the subdivision of $\vertexa \rightarrow \vertexb$ in $\graphG_n$ and $\vertexr$ is $k$ edges from $\vertexc$ along the subdivision of $\vertexc \rightarrow \vertexd$ in $\graphG_n$. Further, let us define ``interpolation coefficients'' associated to $\vertexa, \vertexb, \vertexc$ and $\vertexd$ by
\begin{equation*}
\mu_\vertexa = \frac{n-j}{n}, \quad 
\mu_\vertexb = \frac{j}{n}, \quad
\mu_\vertexc = -\frac{n-k}{n}, \quad
\mu_\vertexd = -\frac{k}{n}.
\end{equation*}
Further, to simplify notation, let us write $\mathcal{S} = \{ \vertexa, \vertexb, \vertexc, \vertexd \}$. We may express the resistance $r_{\vertexq \vertexr}$ in terms of the resistances among $\vertexa, \vertexb, \vertexc$ and $\vertexd$ as follows. If $\vertexa \rightarrow \vertexb$ and $\vertexc \rightarrow \vertexd$ are different edges, 
\begin{equation} 
r_{\vertexq \vertexr} = \frac{j(n-j) + k(n-k)}{n} - \frac{1}{2} \sum_{\vertexv, \vertexw \in \mathcal{S}} \mu_\vertexv \mu_\vertexw \, r_{\vertexv \vertexw}
\label{eq:decomposed resistance}
\end{equation}
while if $\vertexa \rightarrow \vertexb = \vertexc \rightarrow \vertexd$, we modify the right-hand side of~\eqref{eq:decomposed resistance} by adding  
\begin{equation*}
-2 \frac{\min(j,k) (n - \max(j,k))}{n}.
\end{equation*}
\label{prop:resistance formula}
\end{proposition}
We note that this is a generalization of Theorem~4.1 in~\cite{Carmona:2017ge} and Example 6 in~\cite{Chen:2010da} in two ways: our result applies to multigraphs with multiple and loop edges and also treats the cases of subdivisions with any number of inserted vertices.

We can use Proposition~\ref{prop:resistance formula} -- in particular \eqref{eq:decomposed resistance} -- to evaluate various statistical physical quantities for topological polymers in solution such as the hydrodynamic radius. 

\begin{proof}
We first observe that $\sum_{\vertexv \in \mathcal{S}} \mu_v = 0$ (regardless of the values of $j$ and $k$). We will use this observation regularly below. Next, 
let us define a new function $p^\mu := \sum_{\vertexv \in \mathcal{S}} \mu_\vertexv p^{\vertexv}$. If $\vertexv \in \mathcal{S}$ and $\vertexw \not\in \mathcal{S}$, then $L p^{\vertexv}(w) = -\frac{1}{\verticesV}$. Thus, if $\vertexw \not\in \mathcal{S}$, $(L p^\mu)(w) = -\frac{1}{\verticesV} \sum_{\vertexv \in \mathcal{S}} \mu_v = 0$. 

This means that Lemma~\ref{lem:harmonic implies linear} tells us that $p^\mu$ is linear on $\vertexa \rightarrow \vertexb$ and $\vertexc \rightarrow \vertexd$. Thus we can evaluate $p^\mu$ at $\vertexq$ by linearly interpolating values at $\vertexa$ and $\vertexb$ (and likewise for $\vertexr$, $\vertexc$ and $\vertexd$). In particular, keeping track of the signs in $\mu_c$ and $\mu_d$,
\begin{equation}
p^{\mu}(q) - p^{\mu}(r) = \sum_{\vertexw \in \mathcal{S}} \mu_w \, p^\mu(w).
\label{eq:plambda at q minus r}
\end{equation}
Now~\eqref{eq:resistance distance definition} tells us that $r_{\vertexq \vertexr} = p^{\vertexq \vertexr}(\vertexq) - p^{\vertexq \vertexr}(\vertexr)$. Further, using the definition of $p^\mu$, Proposition~\ref{prop:decomposition of solutions to poisson problems} can be rephrased as $p^{\vertexq \vertexr} = p^{\vertexa j \vertexb} - p^{\vertexc k \vertexd} + p^\mu$. Using the definition of $p^{\mu}$ and~\eqref{eq:plambda at q minus r}, we then have 
\begin{equation*}
r_{\vertexq \vertexr} = (p^{\vertexa j \vertexb} - p^{\vertexc k \vertexd})(\vertexq) - (p^{\vertexa j \vertexb} - p^{\vertexc k \vertexd})(r) + \sum_{\vertex,\vertexw \in \mathcal{S}} \mu_\vertex \mu_\vertexw \, p^\vertex(\vertexw).
\end{equation*}
We now consider the right hand side. Using~\eqref{eq:resistance distance definition} we can expand 
\begin{equation*}
\sum_{\vertexv, \vertexw \in \mathcal{S}} \mu_\vertexv \mu_\vertexw \, r_{\vertexv \vertexw} = 
\sum_{\vertexv, \vertexw \in \mathcal{S}} \mu_\vertexv \mu_\vertexw (p^{\vertexv}(\vertexv) + p^{\vertexw}(\vertexw) - p^{\vertexv}(w) - p^{\vertexw}(v)) 
%&= 2 \sum_{\vertexv, \vertexw \in \mathcal{S}} \mu_\vertexv \mu_\vertexw p^{\vertexw}(\vertexw) - 2 \sum_{\vertexv, \vertexw \in \mathcal{S}} \mu_\vertexv \mu_\vertexw p^{\vertexv}(\vertexw) \\
%&= 2 \sum_{\vertexv \in \mathcal{S}} \mu_\vertexv  \sum_{\vertexw \in \mathcal{S}} \mu_\vertexw \, p^{\vertexw}(\vertexw) - 2 \sum_{\vertexv, \vertexw \in \mathcal{S}} \mu_\vertexv \mu_\vertexw \, p^{\vertexv}(\vertexw) \\
= - 2 \sum_{\vertexv, \vertexw \in \mathcal{S}} \mu_\vertexv \mu_\vertexw \, p^{\vertexv}(\vertexw)
\end{equation*}
where we used $\sum_{\vertexv \in \mathcal{S}} \mu_\vertexv = 0$ in the second equality. Now assume that $\vertexa \rightarrow \vertexb$ and $\vertexc \rightarrow \vertexd$ are different. It remains only to prove that 
\begin{equation}
(p^{\vertexa j \vertexb} - p^{\vertexc k \vertexd})(\vertexq) - (p^{\vertexa j \vertexb} - p^{\vertexc k \vertexd})(r) = \frac{j(n-j) + k(n-k)}{n}.
\label{eq:last term}
\end{equation}
Using~\eqref{eq:pajb solution} in Proposition~\ref{prop:pajb}, we can compute that 
\begin{equation*}
p^{\vertexa j \vertexb}(\vertexq) - p^{\vertexa j \vertexb}(\vertexr) = \frac{j(n-j)}{n}
\end{equation*}
since $r$ is not on $\vertexa \rightarrow \vertexb$ so $p^{\vertexa j \vertexb}(\vertexr) = C$. This and the corresponding statement for $p^{\vertexc k \vertexd}(\vertexr) - p^{\vertexc k \vertexd}(\vertexq)$ establish~\eqref{eq:last term}. If $\vertexa \rightarrow \vertexb = \vertexc \rightarrow \vertexd$ the situation is slightly more vexing, since $\vertexr$ is on $\vertexa \rightarrow \vertexb$. In this case, we see that 
\begin{equation*}
p^{\vertexa j \vertexb}(\vertexr) = C + 
\begin{cases}
\frac{k(n-j)}{n}, & \text{if $k \leq j$} \\
\frac{j(n-k)}{n}, & \text{if $k \geq j$}
\end{cases}
= C + \frac{\min(j,k) (n - \max(j,k))}{n}.
\end{equation*}
\end{proof}

\subsection{Computing the Kirchhoff index of $\graphG_n$}

We are now interested in computing $\Kirchhoff(\graphG_n)$. It will be easier to compute $2 \Kirchhoff(\graphG_n) = \sum_{\vertexv, \vertexw \in \graphG_n} r_{\vertexv \vertexw}$. For convenience, we define
\begin{definition}
We call a vertex $\vertexv \in \graphG_n$ an \emph{interior} vertex if it is one of the vertices $\vertex_1, \dots, \vertex_{n-1}$ along the subdivision in $\graphG_n$ of an edge $\vertexa \rightarrow \vertexb$ in $\graphG$ (that is, it is not $\vertex_0 = \vertexa$ or $\vertex_n = \vertexb$). The set of interior vertices will be called $\mathcal{I}$. The interior vertices along $\vertexa \rightarrow \vertexb$ will be $\mathcal{I}(\vertexa \rightarrow \vertexb)$. We call vertices that are not interior~\emph{boundary} vertices, and note that they are also vertices of $\graphG$.
\end{definition}

It is now clear that 
\begin{equation}
2 \Kirchhoff(\graphG_n) = \sum_{\vertexq, \vertexr \in \graphG_n} \hspace{-0.05in} r_{\vertexq \vertexr} = 
\sum_{\vertexq, \vertexr \in \mathcal{I}}  r_{\vertexq \vertexr} + 
2 \sum_{\substack{\vertexq \in \mathcal{I}\\ \vertexr \in \mathcal{B}}} r_{\vertexq \vertexr} + \sum_{\vertexq, \vertexr \in \mathcal{B}} r_{\vertexq \vertexr}.
\label{eq:Kirchhoff structure}
\end{equation}

We are ready to prove Theorem~\ref{thm:degree Kirchhoff formula}, which is really the heart of Theorem~\ref{thm:main}:

\begin{degreeKirchhoff}
If $\graphG$ is a connected multigraph and $\graphG_n$ is the graph obtained by subdividing each edge of $\graphG$ into $n$ pieces, then 
\begin{equation*}
\lim_{n \rightarrow \infty} \frac{1}{n^3} \Kirchhoff(\graphG_n) = \frac{2 \Loops(\graphG) - 1}{12} \edgesE(\graphG) + \frac{1}{4} \DegreeKirchhoff(\graphG)
\end{equation*}
where $\Loops(\graphG) = \edgesE(\graphG) - \verticesV(\graphG) + 1$ is the cycle rank of $\graphG$ and $\DegreeKirchhoff(\graphG)$ is the ``degree-Kirchhoff index'' introduced by Chen and Zhang~\cite{Chen:2007dv}:
\begin{equation*}
\DegreeKirchhoff(\graphG) = \sum_{\vertex_i < \vertex_j \in \graphG} (\deg_\graphG \vertex_i \deg_\graphG \vertex_j ) r_{\vertex_i \vertex_j}.
\end{equation*}
\end{degreeKirchhoff}
We note that this result is compatible with the result of iteratively applying the subdivision operator given by Theorem~3.4 of~\cite{Yang:2015eq}, but that it doesn't follow from it-- iterating the usual subdivision operator can only give $2^k$-fold edge subdivisions of $\graphG$, and we don't know that $\lim_{n \rightarrow \infty} \frac{1}{n^3} \Kirchhoff(\graphG_n) = \lim_{k \rightarrow \infty} \frac{1}{(2^k)^3} \Kirchhoff(\graphG_{2^k})$ until we know that both limits exist. Further,~\cite{Yang:2015eq} doesn't identify the number of loops as part of the result and depends on a result of~\cite{Chen:2010da} which is only proved for graphs without multiple edges.

\begin{proof}
We will repeatedly use the fact that if $\vertexq$ and $\vertexr$ are in $\graphG$ and $\graphG_n$, then $r_{\vertexq \vertexr}^{\graphG_n} = n \, r_{\vertexq \vertexr}^{\graphG}$ because (from the point of view of $\vertexq$ and $\vertexr$) we can regard $\graphG_n$ as the electrical network obtained by replacing each edge of $\graphG$ with a resistor of resistance $n$.

Using Proposition~\ref{prop:resistance formula} we see that if $\vertexa \rightarrow \vertexb$ and $\vertexc \rightarrow \vertexd$ are different edges of $\graphG$, then 
\begin{equation}
\lim_{n \rightarrow \infty} \frac{1}{n^3} \sum_{\substack{\vertexq \in \mathcal{I}(\vertexa \rightarrow \vertexb) \\ \vertexr \in \mathcal{I}(\vertexc \rightarrow \vertexd)}} \hspace{-0.1in} r_{\vertexq \vertexr}^{\graphG_n} = \frac{1}{3} - \frac{1}{6} (r^\graphG_{\vertexa \vertexb} + r^\graphG_{\vertexc \vertexd}) + \frac{1}{4} (r^\graphG_{\vertexa \vertexc} + r^\graphG_{\vertexa \vertexd} + r^\graphG_{\vertexb \vertexc} + r^\graphG_{\vertexb \vertexd}).
\label{eq:two edge case}
\end{equation}
We now need to sum~\eqref{eq:two edge case} over pairs of (different) edges in $\graphG$. We start with
\begin{equation*}
\sum_{\substack{\vertexa \rightarrow \vertexb \neq \vertexc \rightarrow \vertexd\\\text{edges of $\graphG$}}} r^\graphG_{\vertexa \vertexb} + r^\graphG_{\vertexc \vertexd} =
2 \hspace{-0.1in} \sum_{\substack{\vertexa \rightarrow \vertexb \neq \vertexc \rightarrow \vertexd\\\text{edges of $\graphG$}}} \hspace{-0.1in} r^\graphG_{\vertexa \vertexb}
= 2 (\edgesE(\graphG) - 1) \hspace{-0.1in} \sum_{\substack{\vertexa \rightarrow \vertexb \\ \text{edge of $\graphG$}}} \hspace{-0.1in} r^\graphG_{\vertexa \vertexb} 
= 2 (\edgesE(\graphG) - 1)(\verticesV(\graphG) - 1).
\end{equation*}
The last step follows from Foster's theorem, originally proved by Foster~\cite{Foster:1948vk} for graphs and generalized to weighted graphs by Tetali~\cite{Tetali:1991es}. In turn, the weighted graph result implies the result for multigraphs: loop edges contribute nothing to resistance distance, since the resistance distance from any vertex to itself is clearly zero, and, since conductances add in parallel circuits, we can interpret $k$ edges connecting two vertices as a single edge of weight (i.e., conductance) $k$ for the purposes of computing resistance distances. 

We next consider 
\begin{equation*}
\sum_{\substack{\vertexa \rightarrow \vertexb \neq \vertexc \rightarrow \vertexd\\\text{edges of $\graphG$}}} r^\graphG_{\vertexa \vertexc} + r^\graphG_{\vertexa \vertexd} + r^\graphG_{\vertexb \vertexc} + r^\graphG_{\vertexb \vertexd} = 
\sum_{\substack{\vertexa \rightarrow \vertexb \\ \text{edge of $\graphG$}}} \left(
\sum_{\substack{\vertexc \rightarrow \vertexd \text{ edge of} \\ \text{ $\graphG - (\vertexa \rightarrow \vertexb)$}}} 
r^\graphG_{\vertexa \vertexc} + r^\graphG_{\vertexa \vertexd} 
+ r^\graphG_{\vertexb \vertexc} + r^\graphG_{\vertexb \vertexd} \right).
\end{equation*}
Choose any vertex $\vertexr \in \graphG$. In the inner sum, $\vertexr$ appears as $\vertexc$ once for every edge with tail $\vertexr$ in $\graphG - (\vertexa \rightarrow \vertexb)$ and $\vertexr$ appears as $\vertexd$ once for every edge with head $\vertexr$ in $\graphG - (\vertexa \rightarrow \vertexb)$. Note that loop edges pose no special difficulty-- $\vertexr$ may appear as both $\vertexc$ and $\vertexd$ for a single $\vertexc \rightarrow \vertexd$ in this case. Further, multiple edges also require no special treatment. Together, $r_{\vertexa \vertexr}$ and $r_{\vertexb \vertexr}$ appear $\deg_{\graphG - (\vertexa \rightarrow \vertexb)} \vertexr$ times in the inner sum, so 
\begin{equation*}
\sum_{\substack{\vertexa \rightarrow \vertexb \\ \text{edge of $\graphG$}}} \left(
\sum_{\substack{\vertexc \rightarrow \vertexd \text{ edge of} \\ \text{ $\graphG - (\vertexa \rightarrow \vertexb)$}}} 
r^\graphG_{\vertexa \vertexc} + r^\graphG_{\vertexa \vertexd} 
+ r^\graphG_{\vertexb \vertexc} + r^\graphG_{\vertexb \vertexd} \right) = 
\sum_{\substack{\vertexa \rightarrow \vertexb \\ \text{edge of $\graphG$}}} \sum_{\vertexr \in \graphG} (\deg_{\graphG - (\vertexa \rightarrow \vertexb)} r) (r^\graphG_{\vertexa \vertexr} + r^\graphG_{\vertexb \vertexr}).
\end{equation*}
Now $\deg_{\graphG - (\vertexa \rightarrow \vertexb)} r = \deg_\graphG r$ for all $\vertexr$ except $\vertexa$ and $\vertexb$, where we have $\deg_{\graphG - (\vertexa \rightarrow \vertexb)} a = \deg_\graphG a - 1$ and $\deg_{\graphG - (\vertexa \rightarrow \vertexb)} b = \deg_\graphG b - 1$. Thus
\begin{equation*}
\sum_{\substack{\vertexa \rightarrow \vertexb \\ \text{edge of $\graphG$}}} \sum_{\vertexr \in \graphG} (\deg_{\graphG - (\vertexa \rightarrow \vertexb)} r) (r^\graphG_{\vertexa \vertexr} + r^\graphG_{\vertexb \vertexr}) =
\sum_{\substack{\vertexa \rightarrow \vertexb \\ \text{edge of $\graphG$}}} \left( - 2 r_{\vertexa \vertexb} +
 \sum_{\vertexr \in \graphG} (\deg_\graphG \vertexr) (r^\graphG_{\vertexa \vertexr} + r^\graphG_{\vertexb \vertexr}) \right). 
\end{equation*}
The first sum is equal to $-2 (\verticesV(\graphG) - 1)$, again by Foster's theorem. As above, if we fix any $\vertexq \in \graphG$, it appears as $\vertexa$ in the sum once for every edge with tail $\vertexq$ and as $\vertexb$ in the sum above once for every edge with head $\vertexq$. Therefore, the term $r^\graphG_{\vertexq \vertexr}$ appears $\deg_\graphG \vertexq$ times in the sum. We have thus proved
\begin{align*}
\sum_{\substack{\vertexa \rightarrow \vertexb \\ \text{edge of $\graphG$}}} \left( - 2 r_{\vertexa \vertexb} +
 \sum_{\vertexr \in \graphG} (\deg_\graphG \vertexr) (r^\graphG_{\vertexa \vertexr} + r^\graphG_{\vertexb \vertexr}) \right) &=
-2 (\verticesV(\graphG) - 1) + \sum_{\vertexq, \vertexr \in \graphG} (\deg_\graphG \vertexq \deg_\graphG \vertexr ) r_{\vertexq \vertexr} \\
&= -2 (\verticesV(\graphG) - 1) + 2 \DegreeKirchhoff(\graphG).
\end{align*}
Putting all of this together, we can conclude that
\begin{equation*}
\lim_{n \rightarrow \infty} \frac{1}{n^3} \hspace{-0.1in}
\sum_{\vertexa \rightarrow \vertexb \neq \vertexc \rightarrow \vertexd}
\sum_{
\substack{
	\vertexq \in \mathcal{I}(\vertexa \rightarrow \vertexb) \\
	\vertexr \in \mathcal{I}(\vertexc \rightarrow \vertexd) 
}
}
\hspace{-0.1in} r_{\vertexq \vertexr}^{\graphG_n} = 
\frac{1 + 2\edgesE(\graphG) (\edgesE(\graphG) - \verticesV(\graphG)) - \verticesV(\graphG)}{6} + \frac{1}{2} \DegreeKirchhoff(\graphG).
\end{equation*}

Further, if $\vertexa \rightarrow \vertexb$ is any single edge of $\graphG$, then
\begin{equation*}
\lim_{n\rightarrow\infty} \frac{1}{n^3} \sum_{\vertexq, \vertexr \in \mathcal{I}(\vertexa \rightarrow \vertexb)}  r_{\vertexq \vertexr}^{\graphG_n} = \frac{1}{6} + \frac{1}{6} r^\graphG_{\vertexa \vertexb}, \label{eq:single edge case}
\end{equation*}
so we have (using Foster's theorem again),
\begin{equation*}
\lim_{n\rightarrow\infty} \frac{1}{n^3} \sum_{\vertexa \rightarrow \vertexb} \sum_{\vertexq, \vertexr \in \mathcal{I}(\vertexa \rightarrow \vertexb)}  r_{\vertexq \vertexr}^{\graphG_n} = \frac{\edgesE(G) + \verticesV(\graphG)-1}{6}.
\end{equation*}

We are now ready to return to~\eqref{eq:Kirchhoff structure} and assemble the pieces. Recalling that the cycle rank ${\Loops(\graphG) = \edgesE(\graphG) - \verticesV(\graphG) + 1}$, we get
\begin{equation}
\lim_{n \rightarrow \infty} \frac{1}{n^3} \sum_{\vertexq, \vertexr \in \mathcal{I}} r^{\graphG_n}_{\vertexq \vertexr} = \frac{2 \Loops(\graphG) - 1}{6} \edgesE(\graphG) + \frac{1}{2} \DegreeKirchhoff(\graphG).
\label{eq:limiting form 1}
\end{equation}
Using Proposition~\ref{prop:resistance formula} and summing directly, one finds that 
\begin{equation*}
\lim_{n \rightarrow \infty} \frac{1}{n^3}
\left( 2 \sum_{\substack{\vertexq \in \mathcal{I}\\ \vertexr \in \mathcal{B}}} r_{\vertexq \vertexr} + \sum_{\vertexq, \vertexr \in \mathcal{B}} r_{\vertexq \vertexr} \right) = 0,
\end{equation*}
so in fact~\eqref{eq:limiting form 1} proves the theorem.
\end{proof}

\section{Proof of Theorem~\ref{thm:main}}

We are now ready to prove Theorem~\ref{thm:main}. 

\begin{proof}
Recall from Theorem~\ref{thm:expected radius of gyration} that 
\begin{equation*}
\mathcal{E}\left(R_g^2;\graphG_n\right) = \frac{d}{\verticesV(\graphG_n)^2} \Kirchhoff(\graphG_n).
\end{equation*}
Since $\verticesV(\graphG_n) = \edgesE(\graphG)(n-1) + \verticesV(\graphG)$, we have
\begin{equation*}
\lim_{n \rightarrow \infty} \frac{1}{n} \mathcal{E}\left(R_g^2; \graphG_n \right) = \lim_{n \rightarrow \infty} \frac{d}{n \verticesV(\graphG_n)^2} \Kirchhoff(\graphG_n) =  \frac{d}{\edgesE(\graphG)^2} \lim_{n\rightarrow\infty} \frac{1}{n^3}\Kirchhoff(\graphG_n). 
\end{equation*}
Computing the last term with Theorem~\ref{thm:degree Kirchhoff formula}, we see that 
\begin{equation*}
\frac{d}{\edgesE(\graphG)^2} \lim_{n\rightarrow\infty} \frac{1}{n^3}\Kirchhoff(\graphG_n) = 
\frac{d}{2\edgesE(\graphG)} \left( \frac{1}{3} \Loops(\graphG) - \frac{1}{6} \right) + \frac{d}{4\edgesE(\graphG)^2} \DegreeKirchhoff(\graphG).
\end{equation*}
A result of Chen and Zhang~\cite{Chen:2007gi}, generalized to weighted graphs and hence to multigraphs by Chen~\cite[p.~1694]{Chen:2010da}, is that the degree-Kirchhoff index is related to the normalized graph Laplacian by
\begin{equation*}
\DegreeKirchhoff(\graphG) =  2 \edgesE(\graphG) \tr \mathcal{L}^+,
\end{equation*}
which completes the proof.
\end{proof}

\section{Contraction factors}

We can now study the asymptotic behavior of the expected radius of gyration in a subdivided graph. But it's still a little unclear why we should divide by $n$, and what these coefficients really mean from a polymer science point of view. We can give a nicer interpretation of our theorem by recalling the idea of the contraction factor (cf.~\cite[1.48]{stepto:2015be}). 

\begin{definition}
The \emph{contraction factor} $g$ of a TCRW with graph $\graphG$ and $\verticesV$ vertices is the ratio 
\begin{equation*}
g(\graphG) = \frac{\mathcal{E}(R_g^2; \graphG)}{\mathcal{E}(R_g^2; \text{ path graph with $\verticesV$ vertices})}.
\end{equation*}
\end{definition}

Combining Theorem~\ref{thm:main} with the formula~\eqref{eq:path gyradius} for the expected radius of gyration of the path graph yields Theorem~\ref{thm:contraction factors}, which we restate for convenience. 

\begin{gFactor}
For any connected multigraph $\graphG$ (including loop and multiple edges), if $\graphG_n$ is the $n$-part edge subdivision of $\graphG$, then the contraction factor $g(\graphG_n)$ of a Gaussian TCRW with underlying graph $\graphG_n$ embedded in $\R^d$ obeys
\begin{equation*}
g(\graphG_\infty) := \lim_{n \rightarrow \infty} g(\graphG_n) = \frac{3}{\edgesE(\graphG)^2} \left( \operatorname{Tr} \mathcal{L}^+(\graphG) + \frac{1}{3} \operatorname{Loops}(\graphG) - \frac{1}{6} \right).
\end{equation*}
\end{gFactor}

We may now use Theorem~\ref{thm:contraction factors} to estimate the relative sizes in solution of large topological polymers with different underlying graphs. Figure~\ref{fig:example} shows an example of a topological polymer synthesized in the Tezuka lab. A direct computation shows that the eigenvalues of the $6 \times 6$ matrix $\mathcal{L}(\graphG)$ are $2, \frac{5}{3}, 1, 1, \frac{1}{3}$ and $0$. This means that the eigenvalues of $\mathcal{L}^+(\graphG)$ are $\frac{1}{2}, \frac{3}{5}, 1, 1, 3$ and $0$. Since $\Loops(\graphG) = 4$ the result of Theorem~\ref{thm:contraction factors} is that
\begin{equation*}
g(\graphG_\infty) = \frac{109}{405} \simeq 0.269135.
\end{equation*} 
A quick numerical experiment shows that for the graph $\graphG$ in Figure~\ref{fig:example}, 
$g(\graphG_{10}) \simeq 0.270064$ (sample average over 1 million trials), which shows that Theorem~\ref{thm:contraction factors} is a quite good estimate even for small values of $n$. In fact, $g(\graphG_{2}) \simeq 0.261687$ (sample average over 1 million trials), so the estimate is useful even for $n=2$.

\begin{definition}
The \emph{relative contraction factor} of $\graphG_1$ and $\graphG_2$ is given by
\begin{equation*}
g(\graphG_1,\graphG_2) = \frac{\mathcal{E}(R_g^2; \graphG_1)}{\mathcal{E}(R_g^2; \graphG_2)} = \frac{g(\graphG_1)}{g(\graphG_2)}.
\end{equation*}
\end{definition}

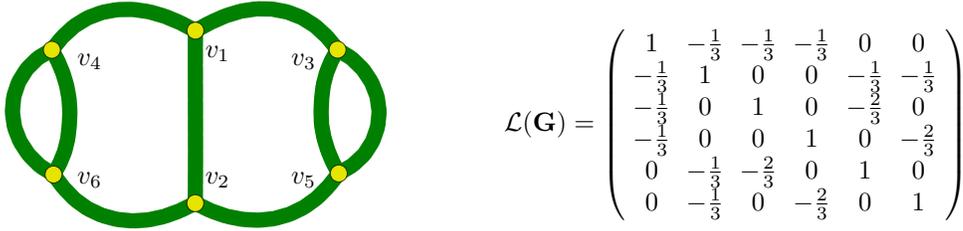
\begin{figure}
\hfill
\raisebox{-0.5\height}{
\begin{overpic}[width=2in]{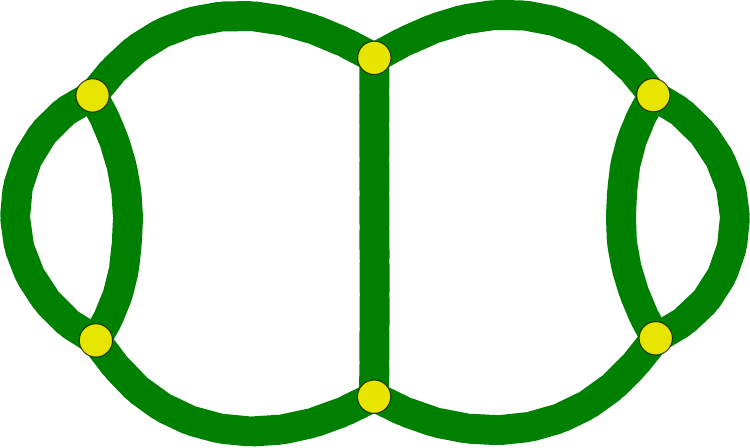}
\put(52.5,45){$\vertex_1$}
\put(52.5,12){$\vertex_2$}
\put(19,43){$\vertex_4$}
\put(19,12){$\vertex_6$}
\put(75,12){$\vertex_5$}
\put(75,43){$\vertex_3$}
\end{overpic}
}
\hfill
\begin{minipage}{3in}
\begin{equation*}
\mathcal{L}(\graphG) = 
\left(
\begin{array}{cccccc}
 1 & -\frac{1}{3} & -\frac{1}{3} &
   -\frac{1}{3} & 0 & 0 \\
 -\frac{1}{3} & 1 & 0 & 0 &
   -\frac{1}{3} & -\frac{1}{3} \\
 -\frac{1}{3} & 0 & 1 & 0 &
   -\frac{2}{3} & 0 \\
 -\frac{1}{3} & 0 & 0 & 1 & 0 &
   -\frac{2}{3} \\
 0 & -\frac{1}{3} & -\frac{2}{3} & 0
   & 1 & 0 \\
 0 & -\frac{1}{3} & 0 & -\frac{2}{3}
   & 0 & 1 \\
\end{array}
\right)
\label{eq:ngl}
\end{equation*}
\end{minipage}
\hfill
\hphantom{.}
\caption{A multigraph topological polymer synthesized in the Tezuka lab~\cite{Suzuki:2014fo} together with its normalized graph Laplacian $\mathcal{L}(\graphG)$.}
\label{fig:example}
\end{figure}

We now consider the example of the family of tricyclic and tetracyclic topological polymers shown in Figure~\ref{fig:predictions}. We chose these structure graphs because these polymers were actually synthesized by Suzuki, Yamamoto, and Tezuka~\cite{Suzuki:2014fo} in 2014. To compare our simple model to a more detailed simulation, we constructed random configurations with topological constraints and self-avoiding effects via LAMMPS with the TSUBAME supercomputer at the Tokyo Institute of Technology. As is standard for simulations in the Kremer-Grest model, we used a repulsive Lennard-Jones (LJ) potential to model steric effects and a finitely extensive and nonlinear elastic (FENE) potential to model the bonds between atoms.

We took conformations at every $5\tau$ steps, where the relaxation time $\tau$ was defined by the number of time steps at which the correlation between conformations was $1/e$. This relaxation time depends on the graph type of polymers and increases as the number of atoms grows. For instance, for the bipartite complete graph $K_{3,3}$ of $447$ atoms $\tau = 3.2 \times 10^5$. For the same graph with $897$ atoms, $\tau = 1.6\times 10^6$ steps. For the alpha graph with comparable numbers of atoms relaxation took longer, requiring  $\tau = 5.3\times 10^5$ for a $448$ atom graph and $2.0\times 10^6$ steps for an $898$-atom graph.

For polymer chains which have excluded volume, the expected squared radius of gyration should be proportional to $\verticesV^{2 \times 0.588}$ for large $\verticesV$. Thus we fitted the results of the molecular dynamics data to $\mathcal{E}(R_g^2;\graphG_n) = C_{\graphG} \verticesV(\graphG_n)^{1.176} + \Delta_{\graphG}$. We measured the quality of the fit to this model by $\chi^2/DF$ (degrees of freedom), and obtained values $\leq 1.1$ for all the graph types tested. We were then able to compute asymptotic contraction factors for our molecular dynamics data relative to the tree graph $\graphG^\text{tree}$ by taking 
\begin{equation}
g(\graphG_\infty,\graphG^\text{tree}_\infty)^\text{MD} = \frac{C_\graphG}{C_{\text{tree}}}.
\label{eq:lammps relative contraction}
\end{equation}
The resulting estimates of $g(\graphG_\infty,\graphG^\text{tree}_\infty)^\text{MD}$ are given in the table at left in~Figure~\ref{fig:predictions}, which also compares these results to the predictions of Theorem~\ref{thm:contraction factors}. We see that the extremely simple computation in the theorem, which requires only finding the eigenvalues of a small matrix, is very successful at predicting the results of molecular dynamics simulations which consumed dozens of hours of supercomputer time.

\begin{figure}
\hfill
\begin{tabular}{cccc} \hline \hline
 $\graphG$ & $g(\graphG_\infty,\graphG^\text{tree}_\infty)^\text{MD}$ &  $g(\graphG_\infty,\graphG^\text{tree}_\infty)$ \\ \hline
\raisebox{-0.5\height}{\includegraphics[height=0.25in]{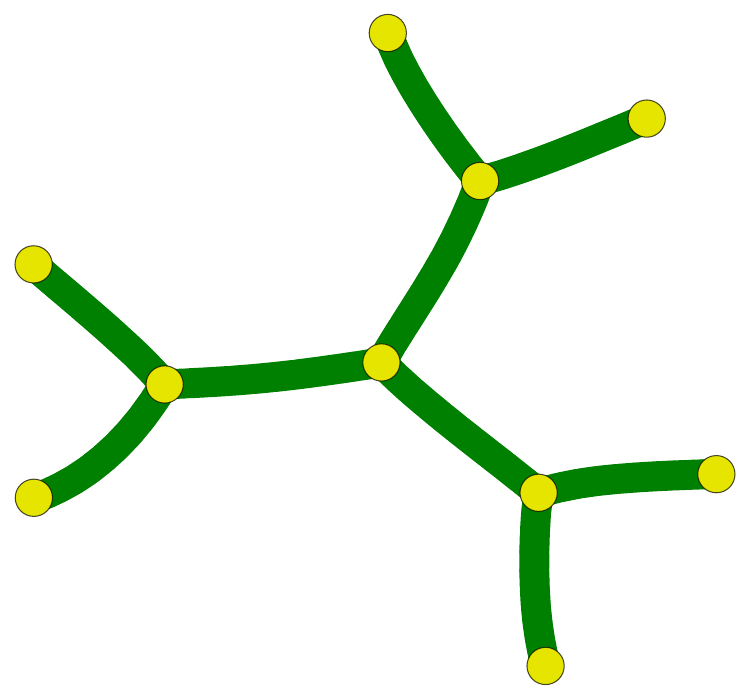}} & $1.0$ &  1 \\
\raisebox{-0.5\height}{\includegraphics[height=0.25in]{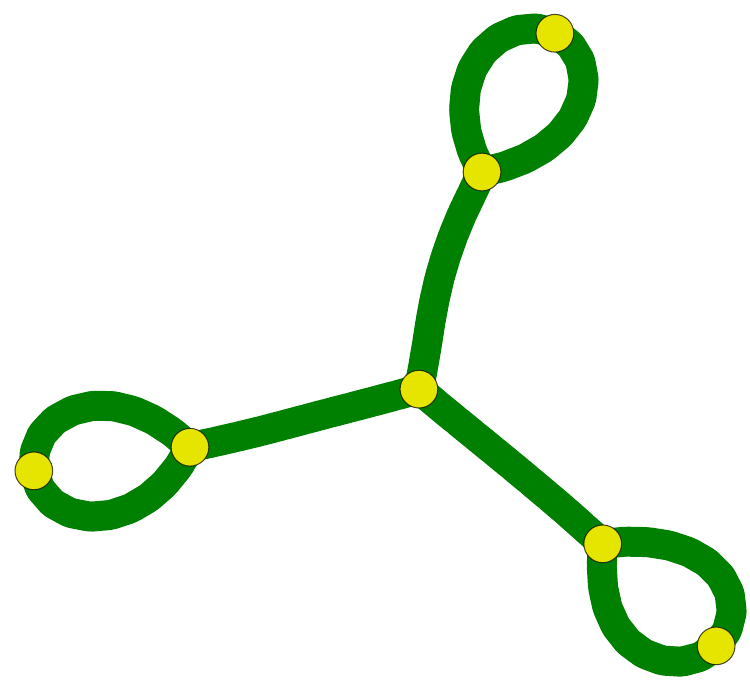}} & $0.962 \pm 0.034$ & $\nicefrac{43}{49}$ \\
 \raisebox{-0.5\height}{\includegraphics[height=0.25in]{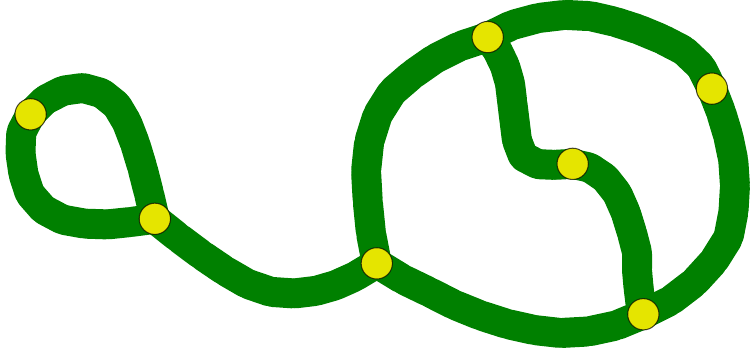}} & $0.782 \pm 0.026$ & $\nicefrac{31}{49}$ \\
 \raisebox{-0.5\height}{\includegraphics[height=0.25in]{tezuka-d.pdf}} & $0.582 \pm 0.015$ & $\nicefrac{109}{245}$  \\
 \raisebox{-0.5\height}{\includegraphics[height=0.25in]{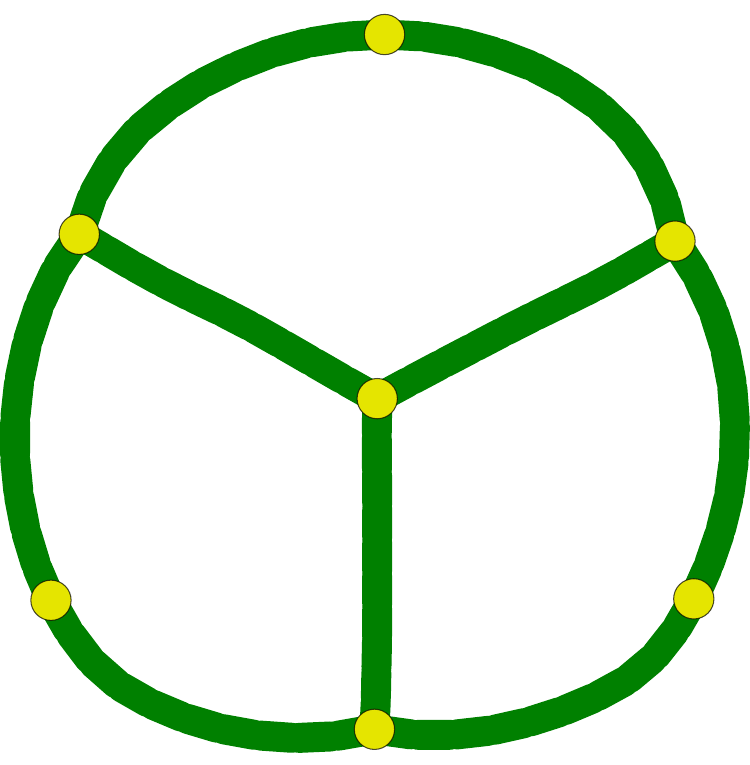}} & $0.546 \pm 0.016$ & $\nicefrac{107}{245}$  \\
 \raisebox{-0.5\height}{\includegraphics[height=0.25in]{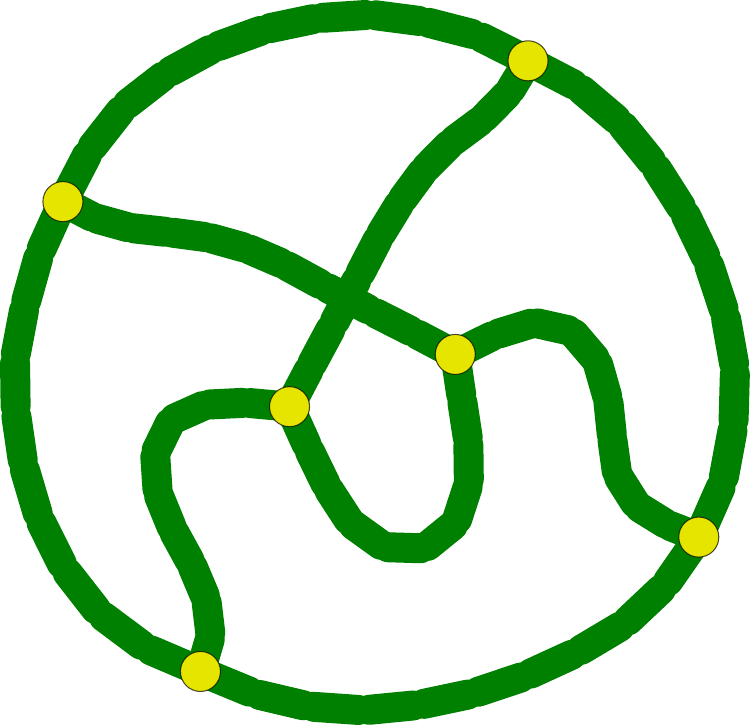}} & $0.445 \pm 0.011$  & $\nicefrac{17}{49}$  \\ \hline \hline
\end{tabular}
\hfill
\raisebox{-0.5\height}{
\begin{overpic}[width=2.5in]{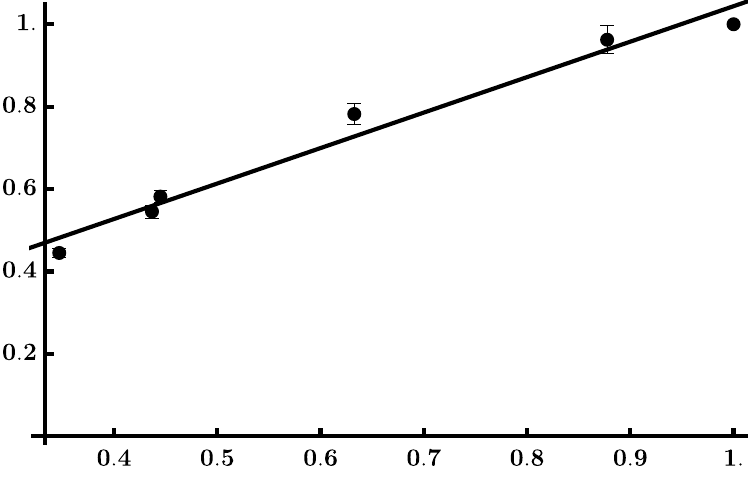}
\put(22,-5){Contraction Factor (Thm.~\ref{thm:contraction factors})}
\put(0,65){MD}
\end{overpic}}
\hfill
\hphantom{.}
\caption{The table above left shows six tricyclic and tetracyclic topological polymers synthesized by Suzuki, Yamamoto, and Tezuka~\cite{Suzuki:2014fo} with the asymptotic relative contraction factor estimated by molecular dynamics simulation using~\eqref{eq:lammps relative contraction} and the asymptotic relative contraction factor predicted by Theorem~\ref{thm:contraction factors}. The plot above right shows the molecular dynamics results ($y$-axis) versus the theoretical relative contraction factors ($x$-axis). The linear fit between molecular dynamics results and~Theorem~\ref{thm:contraction factors} has an $R^2$ (coefficient of determination) of $0.966$.
}
\label{fig:predictions}
\end{figure}

Although the fit in Figure~\ref{fig:predictions} is purely empirical, it is good enough that we might hope that a similar correlation will also hold for other topological polymers with complex graphs. Since estimates of the $g$-factor are obtained in molecular dynamics with excluded volume, while the contraction factors are derived from the ideal chain model of topological polymers with no excluded volume, this would give a dramatic reduction in computational complexity of estimating $g$-factors.  One reason for caution is that in all of our examples the functionality at each vertex is no more than three, so the mean-square radius of gyration may not be much enhanced by the excluded volume effect~\cite{uehara_statistical_2016}. 

\section{Conclusion and Future Directions}

Our approach to phantom network theory is conceptually quite simple: the key idea is that each coordinate vector $w^k$ of the edge vector $w$ is sampled from a multivariate Gaussian \emph{supported on the linear subspace $\im B^T$} of $\R^{\edgesE}$. However, we have seen that despite its simplicity, the model has a rich mathematical theory and significant explanatory power. 

There are many more interesting questions associated to this model. First, it would obviously be interesting to make the model more realistic by incorporating some notion of excluded volume. We note that this can't change the situation too much (at least for large configurations) as our molecular dynamics simulation did include these effects, but Theorem~\ref{thm:contraction factors} still explained the simulation data quite well. Second, it would be interesting to predict the expected value of other quantities used for proxies in SEC, such as the mean width or hydrodynamic radius~(cf. \cite{Wang:2010iz}). However, we think that one of the most promising future directions for study is the prediction of the properties of topological polymers whose underlying graph type is random, such as rubber or collagen. We have taken great care above to build a theory which applies to multigraphs for precisely this reason, as multiple and loop edges can't be excluded from most random graph models without significant difficulty.

\section{Acknowledgments}
We have benefitted tremendously from discussions with various colleagues and friends, and we would especially like to thank Yasuyuki~Tezuka and Satoshi Honda for conversations about topological polymers and Fan Chung for introducing us to spectral graph theory. We are also grateful to the Japan Science and Technology Agency, which sponsored conferences at Ochanomizu University and the Tokyo Institute of Technology that catalyzed the present work, and to the Simons Foundation (\#524120 to Cantarella, \#354225 to Shonkwiler), the Japan Science and Technology Agency (CREST Grant Number JPMJCR19T4), and the Japan Society for the Promotion of Science (KAKENHI Grant Number JP17H06463) for their support.

\bibliography{tcrw-cites,tcrwpapers-special,tcrwpapers}

\end{document}